\documentclass[11pt, oneside]{article}
\usepackage{geometry}
\usepackage{xcolor}
\usepackage{graphicx}
\usepackage{comment}
\usepackage{amssymb,amsmath,amsthm,mathtools}
\usepackage{algorithm,algorithmic}

\title{Small Extended Formulation for Knapsack Cover Inequalities from Monotone Circuits}
\author{Abbas Bazzi \\ EPFL, Switzerland \\ abbas.bazzi@epfl.ch
\and Samuel Fiorini \\ Universit\'e libre de Bruxelles, Belgium \\ sfiorini@ulb.ac.be
\and Sangxia Huang \\ EPFL, Switzerland \\ huang.sangxia@gmail.com
\and Ola Svensson \\ EPFL, Switzerland \\ola.svensson@epfl.ch}

\newcommand{\R}{\mathbb{R}}
\newcommand{\N}{\mathbb{N}}
\newcommand{\Z}{\mathbb{Z}}
\newcommand{\conv}{\mathrm{conv}}

\newcommand{\xc}{\mathrm{xc}}
\newcommand{\nnegrk}{\mathrm{rk}_+}
\newcommand{\CCexp}{R^\mathrm{cc}_{\mathrm{exp}}}
\newcommand{\CCKW}{D^{\mathrm{cc}}_{\mathrm{mon-KW}}}

\newcommand{\demand}{D}
\newcommand{\rdemand}{U}
\newcommand{\size}{s}
\newcommand{\rsize}{s'}
\newcommand{\largeitems}{I_{\mathrm{large}}}
\newcommand{\smallitems}{I_{\mathrm{small}}}
\newcommand{\apxrdemand}{\widetilde{\rdemand}}
\newcommand{\apxsmallcontrib}{\widetilde{\sigma}}
\newcommand{\depth}{t}

\newcommand{\Tconstr}{T_\mathrm{constr}}
\newcommand{\Tsolve}{T_\mathrm{solve}}
\newcommand{\Tsep}{T_\mathrm{sep}}
\newcommand{\preach}{p_\mathrm{reach}}
\newcommand{\pbranch}{p_\mathrm{branch}}
\newcommand{\qbranch}{q_\mathrm{branch}}
\newcommand{\xsize}{x}

\newtheorem{theorem}{Theorem}
\newtheorem{Definition}{Definition}
\newtheorem{Lemma}[theorem]{Lemma}

\begin{document}
\begin{titlepage} 
\maketitle
\thispagestyle{empty}

\begin{abstract}
Initially developed for the min-knapsack problem, the knapsack cover
inequalities are used in the current best relaxations for numerous 
combinatorial optimization problems of covering type. In spite of 
their widespread use, these inequalities yield linear programming (LP)
relaxations of exponential size, over which it is not known how to 
optimize exactly in polynomial time. In this paper we address this 
issue and obtain LP relaxations of quasi-polynomial size that are at
least as strong as that given by the knapsack cover inequalities.

For the min-knapsack cover problem, our main result can be stated 
formally as follows: for any $\varepsilon >0$, there is a
$(1/\varepsilon)^{O(1)}n^{O(\log n)}$-size LP relaxation with 
an integrality gap of at most $2+\varepsilon$, where $n$ is the
number of items. Prior to this work, there was no known relaxation 
of subexponential size with a constant upper bound on the integrality 
gap. 

Our construction is inspired by a connection between extended
formulations and monotone circuit complexity via Karchmer-Wigderson games.
In particular, our LP is based on $O(\log^2 n)$-depth monotone circuits with
fan-in~$2$ for evaluating weighted threshold functions with $n$ inputs, as
constructed by Beimel and Weinreb. We believe that a further understanding 
of this connection may lead to more positive results complementing the 
numerous lower bounds recently proved for extended formulations. 
\end{abstract}
\end{titlepage}

\section{Introduction} \label{sec:intro}
\emph{Capacitated covering problems}\footnote{The term ``capacitated'' is used
in the literature to emphasize that the entries of matrix $A$ can take any
non-negative value in contrast to the uncapacitated version where entries are
Boolean.}  play a central role in combinatorial optimization. These are the
problems modeled by Integer Programs (IPs) of the form  $\min\{\sum_{i=1}^n
c_i x_i \mid Ax \geqslant b,\ x \in \{0,1\}^n\}$, where $A$ is a size-$m \times
n$ nonnegative matrix and $b, c$ size-$n$ nonnegative vectors. The
\emph{min-knapsack problem} is the special case arising when there is a single
covering constraint, that is, when $m = 1$. This is arguably the simplest
interesting capacitated covering problem. 

In terms of complexity, the min-knapsack problem is well-understood: on the one
hand it is weakly NP-hard~\cite{Karp72} and on the other hand it admits an
FPTAS~\cite{Lawler79,PVW85}. However, for its own sake and since it appears as
a key substructure of numerous other IPs, improving our polyhedral
understanding of the problem is important. By this, we mean finding ``good'' linear programming (LP)
relaxations for the min-knapsack problem. Indeed, the polyhedral study of this
problem has led to the development of important tools, such as the
knapsack cover inequalities, for the strengthening of LP
relaxations. These inequalities and generalizations thereof are now used in the current best known
relaxations for several combinatorial optimization problems, such as
single-machine scheduling~\cite{BP10} and capacitated facility location~\cite{AnSS14}. However, despite
this important progress in the past, many fundamental questions remain open
even in the most basic setting. 

\paragraph{State of the Art.} The feasible region of a min-knapsack instance is
specified by positive \emph{item sizes} $\size_1, \ldots, \size_n$ and
a positive \emph{demand} $D$. In this context, a vector $x \in \{0,1\}^n$ is \emph{feasible} if
$\sum_{i=1}^n \size_i x_i \geqslant \demand$. To specify completely an instance
of the min-knapsack problem, we are further given nonnegative \emph{item costs}
$c_1, \ldots, c_n$. Solving the resulting instance then amounts to solving the
IP $\min \{\sum_{i=1}^n c_i x_i \mid \sum_{i=1}^n \size_i x_i \geqslant
\demand,\ x \in \{0,1\}^n \}$.

The \emph{basic} LP relaxation, i.e.,  $\min \{\sum_{i=1}^n c_i x_i \mid \sum_{i=1}^n
\size_i x_i \geqslant \demand,\ x \in [0,1]^n\}$,  provides an estimate on the
optimum value that can be quite bad. More precisely, defining the
\emph{integrality gap} as the supremum over all instances of the ratio of the
optimum value of the IP to the optimum value of the LP relaxation, it is easy
to see that the integrality gap is unbounded. 

Several inequalities have been proposed for strengthening this basic LP
relaxation.  Already in the 70's, Balas~\cite{Balas75}, Hammer, Johnson and
Peled~\cite{HJP75} and Wolsey~\cite{Wolsey75} independently proposed to add the
\emph{uncapacitated} knapsack cover inequalities: for every subset $A \subseteq
[n]$ of the items such that $\sum_{i\in A} s_i < D$, add the inequality
$\sum_{i\not \in A} x_i \geqslant 1$ (saying that at least one item in $[n] \setminus
A$ needs to be picked in order to satisfy the demand). Unfortunately, these
(exponentially many) inequalities are not sufficient for bringing down the
integrality gap to a constant.
A strengthening of these inequalities was therefore proposed more recently by
Carr, Fleischer, Leung and Philipps~\cite{CFLP00}. They defined  
the following valid inequalities: for every set of items $A \subseteq
[n] := \{1,\ldots,n\}$ such that $\sum_{i \in A} \size_i < D$, there is
a corresponding (capacitated) \emph{knapsack cover inequality}
\begin{equation}
\label{eq:KC}
\sum_{i \notin A} \rsize_i x_i \geqslant \rdemand 
\end{equation}
where $\rdemand = \rdemand(A) := D - \sum_{i \in A} \size_i$ is the \emph{residual demand} and $\rsize_i = \rsize_i(A) := \min \{\size_i, \rdemand\}$. The validity of~(\ref{eq:KC}) is due to the fact that every feasible solution $x \in \{0,1\}^n$ has to contain some object $i \notin A$. This object can be \emph{large}, that is, have $\size_i \geqslant \rdemand$, and in this case the inequality is clearly satisfied. Otherwise, in case every object $i \notin A$ is \emph{small}, 
the total size of the objects $i \notin A$ picked by $x$ has to be at least the residual demand $\rdemand$. 

Carr \emph{et al.}~\cite{CFLP00} proved that whenever $x \in \R^n_{\geqslant 0}$ satisfies all knapsack cover inequalities, $2x$ dominates a convex combination of feasible solutions, that is, there exist feasible solutions $x^{(j)} \in \{0,1\}^n$ ($j \in [q]$) and coefficients $\lambda_j \geqslant 0$ summing up to $1$ such that $2x \geqslant \sum_{j=1}^q \lambda_j x^{(j)}$. Given any nonnegative item costs, one of the $x^{(j)}$ will have a cost that is at most $2$ times that of $x$. This implies that the integrality gap of the corresponding LP relaxation is at most~$2$.

The LP relaxation defined by the knapsack cover inequalities is ``good'' in the
sense that it has a constant integrality gap. However, it has exponential
\emph{size}, that is, exponentially many inequalities, over which it is not
known how to optimize exactly in polynomial time; in particular, it is not known how to employ the Ellipsoid algorithm because the problem of separating the knapsack cover inequalities reduces to another knapsack problem (which is NP-hard in general).

In contrast, for the \emph{max-knapsack problem}, Bienstock~\cite{Bienstock08}
proved that for all $\varepsilon > 0$ there exists
a size-$n^{O(1/\varepsilon)}$ LP relaxation whose integrality gap\footnote{For
maximization problems, one takes the supremum of the ratio of the optimum value
of the LP relaxation to the optimum value of the IP.} is at most $1
+ \varepsilon$. That LP is defined by an extended formulation that uses
$n^{O(1/\varepsilon)}$ extra variables besides the $x$-variables. We remark
that it is a notorious open problem to prove or disprove the existence of
a $f(1/\varepsilon) \cdot n^{O(1)}$-size LP relaxation for max-knapsack with
integrality gap at most $1 + \varepsilon$, see e.g.~the survey on extended
formulations by Conforti, Cornu\'ejols and
Zambelli~\cite{ConfortiCornuejolsZambelli10}. Coming back to the min-knapsack
problem, it is not known whether there exists a polynomial-size LP relaxation
with constant integrality gap or not.\footnote{We remark that Bienstock and
  McClosky~\cite{Bienstock2012} considered the easier case when the relaxation is
allowed to depend on the objective function to be optimized (i.e., on the cost
of the items). In this case, using techniques similar to those developed for
polynomial time approximation schemes, they obtained polynomial size
relaxations with integrality gap at most $1+\varepsilon$, for any fixed
$\varepsilon>0$. This is, however, a very different setting and, as the developed inequalities depend on the objective function, they do not generalize to other problems.} 

\paragraph{Main Result.} We come close to resolving the question and show that
min-knapsack admits a quasi-polynomial-size LP relaxation with integrality gap
at most $2 + \varepsilon$. The upper bound on the integrality gap originates
from the fact that our LP relaxation is at least as strong as that provided by
a slightly weakened form of the knapsack cover inequalities. We point out that,
under some conditions, we can bound the size of our relaxation by a polynomial,
see Section~\ref{sec:sub:protocol}. A more precise statement of our main result
is as follows.

\begin{theorem} \label{thm:main}
For all $\varepsilon \in (0,1)$, item sizes $\size_1, \ldots, \size_n \in \R_+$
and demand $D \in \R_+$, there exists a size-$(1/\varepsilon)^{O(1)} n^{O(\log
n)}$ extended formulation defining an LP relaxation of min-knapsack with
integrality gap at most $2+\varepsilon$.
\end{theorem}
As the result is obtained by giving quasi-polynomially many inequalities  of
roughly the same strength as the exponentially many knapsack cover inequalities, our
techniques also lead to relaxations of quasi-polynomial size for the numerous
applications of these inequalities. We mention some of these applications below
when we discuss related works.

Beyond the result itself, the novelty of our approach lies in the concepts we
rely on and the techniques we develop. Our starting point is a connection
between monotone circuits and extended formulations that we explain below. This
connection was instrumental in the recent \emph{lower bounds} of G\"o\"os, Jain
and Watson on the extension complexity of independent set
polytopes~\cite{GJW16}, and can be traced back to a paper of
Hrube\v{s}~\cite{Hrubes12}. Here we use it for the first time to prove an
\emph{upper bound}. 

\paragraph{From Monotone Circuits to Extended Formulations.} Each choice of item sizes and demand gives rise to a \emph{weighted threshold function} $f : \{0,1\}^n \to \{0,1\}$ defined as
\begin{equation}
\label{eq:wtf}
f(x) := 
\begin{cases}
1 &\text{if } \sum_{i=1}^n \size_i x_i \geqslant \demand\\
0 &\text{otherwise}.
\end{cases}
\end{equation}
Since we assume that the item sizes and demand are nonnegative, $f$ is \emph{monotone} in the sense that $a \leqslant b$ implies $f(a) \leqslant f(b)$, for all $a, b \in \{0,1\}^n$. 

Clearly, we have that $x \in \{0,1\}^n$ is feasible if and only if $x \in f^{-1}(1)$. Furthermore, for $a \in f^{-1}(0)$, we can rewrite the uncapacitated knapsack cover inequalities as $\sum_{i : a_i = 0} x_i \geqslant 1$. Consider the \emph{slack matrix} $S_{a,b} := \sum_{i : a_i = 0} b_i - 1$ indexed by pairs $(a,b) \in f^{-1}(0) \times f^{-1}(1)$. By Yannakakis' factorization theorem~\cite{Yannakakis91}, the existence of a size-$r$ LP relaxation of min-knapsack that is at least as strong as that given by the uncapacitated knapsack cover inequalities is equivalent to the existence of a decomposition of the slack matrix $S$ as a sum of $r$ nonnegative rank-$1$ matrices.

Now suppose that there exists a depth-$\depth$ monotone circuit (that is, using only AND gates and OR gates) of fan-in $2$ for computing $f(x)$. A result of Karchmer and Wigderson~\cite{KM90} then implies a partition of the entries of $S$ into at most $2^{\depth}$ rectangles\footnote{A \emph{rectangle} is the Cartesian product of a set of row indices and a set of column indices.} $R \subseteq f^{-1}(0) \times f^{-1}(1)$ such that in each one of these rectangles $R$, there exists some index $i^* = i^*(R)$ such that $a_{i^*} = 0$ and $b_{i^*} = 1$ for all $(a,b) \in R$. Then we may write, for $(a,b) \in R$,
\begin{equation}
\label{eq:slack_decomp_basic}
S_{a,b} = \sum_{i : a_i = 0} b_i - 1
= \sum_{i : a_i = 0,\ i \neq i^*} b_i
= \sum_{i \neq i^*} (1-a_i) b_i
\end{equation}
so that $S$ restricted to the entries of $R$ can be expressed as a sum of at most $n - 1$ nonnegative rank-$1$ matrices of the form $\left((1-a_i) b_i\right)_{(a,b) \in R}$, where $i$ is a fixed index distinct from $i^*$. This implies a decomposition of the whole slack matrix $S$ as a sum of at most $2^{\depth} (n-1)$ nonnegative rank-$1$ matrices, and thus the existence  of a $2^{\depth}(n-1)$-size LP relaxation of min-knapsack that captures the uncapacitated knapsack cover inequalities. Since $f$ is a weighted threshold function, we can take $\depth = O(\log^2 n)$, as proved by Beimel and Weinreb~\cite{BW06}. Therefore, we obtain a $n^{O(\log n)}$-size extended formulation for the uncapacitated knapsack cover inequalities. Unfortunately, these inequalities do not suffice to guarantee a bounded integrality gap.

For the full-fledged knapsack cover inequalities~(\ref{eq:KC}), the simple idea described above breaks down. If the special index $i^* = i^*(R)$ for some rectangle $R$ corresponds to a large object, we can write
$$
\sum_{i : a_i = 0} \rsize_i b_i - \rdemand
= \sum_{i : a_i = 0,\ i \neq i^*} \rsize_i b_i
= \sum_{i \neq i^*} \rsize_i (1-a_i) b_i
$$
where each matrix $\left(\rsize_i (1-a_i) b_i\right)_{(a,b) \in R}$ has rank at most $1$ because $\rsize_i (1-a_i)$ depends on $a$ only. However, $i^*$ may correspond to a small object, in which case we cannot decompose the slack matrix as above.

Nevertheless, we prove that it is possible to overcome this difficulty. Two key ideas we use to achieve this are to discretize some of the quantities (which explains why we lose an $\varepsilon$ in the integrality gap) and resort to several weighted threshold functions instead of just one. If all these functions admit $O(\log n)$-depth monotone circuits of fan-in $2$, then we obtain a size-$n^{O(1)}$ LP relaxation.

\paragraph{Related Works.} Knapsack cover inequalities and their
generalizations such as flow cover inequalities were used as a systematic way
to \emph{strengthen} LP formulations of other (seemingly unrelated)
problems~\cite{CFLP00,CS08,LLS08,BBN08,BGK10,CGK10,BP10,CS11,Efs15}. By
strengthening we mean that one would start with a polynomial size LP
formulation with a potentially unbounded integrality gap for some problem of
interest, and then show that adding (adaptations) of knapsack cover inequalities 
reduces this integrality gap (we illustrate in Section~\ref{sec:FCI} how this
strengthening works for the Single Demand Facility Location problem, reducing
the integrality gap down to 2). However, similar to the case of min-knapsack
discussed above, the drawback of this approach is that the size of the
resulting LP formulation becomes exponential. We can extend our result to show
that it yields quasi-polynomial size LP formulation for many such applications.
To name a few: 

\begin{itemize}
\item Carr \emph{et al.}~\cite{CFLP00} applied these inequalities to the
  Generalized Vertex Cover problem, Multi-color Network Design problem and the
  Fixed Charge Flow problem, and showed how these inequalities reduce the
  integrality gap of the starting LP formulations.

\item Bansal and Pruhs~\cite{BP10} studied the Generalized Scheduling Problem
  (GSP) that captures many interesting scheduling problems such as Weighted
  Flow Time, Flow Time Squared and Weighted Tardiness. In particular, they
  showed a connection between GSP and a certain geometric covering problem, and
  designed an LP based approximation algorithm for the later that yields an
  approximate solution for the GSP. The LP formulation that they use for the
  intermediate geometric cover problem is strengthened using knapsack cover
  inequalities, and yields an $O(\log \log nP)$-approximation for the GSW where
  $n$ is the number of jobs, and $P$ is the maximum job size. In the special
  case of identical release time of the jobs, their LP formulation yields
  a $16$-approximation algorithm. This constant factor approximation was later
  improved by Cheung and Shmoys~\cite{CS11} and Mestre and Verschae~\cite{MV14} to
  a $(4+\varepsilon)$-approximation,
  where the authors added the knapsack cover inequalities directly to the LP
  formulation of the scheduling problem, i.e., without resorting to the
  intermediate geometric cover problem as in~\cite{BP10}. For both the
  GSP and its special case, our method yields an LP formulation whose size is
quasi-polynomial in $n$, and polynomial in both $\log P$ and $\log W$, where
$W$ is the maximum increase in the cost function of a job at any point in
time.

\item Efsandiari \emph{et al.}~\cite{Efs15} used a knapsack-cover-strengthened
  LP formulation to design an $O(\log k)$-approxima\-tion algorithm for
  Precedence-Constrained Single-Machine Deadline scheduling problem, where $k$
  is the number of distinct deadlines. 

\item Carnes and Shmoys~\cite{CS08} designed primal-dual algorithms for the
  Single-Demand Facility Location, where the primal LP formulation is
  strengthened by adding (generalizations) of knapsack cover inequalities.
\end{itemize}

Extended formulations have received a considerable amount of attention recently, mostly for proving impossibility results. Pokutta and Van Vyve~\cite{VP2013} proved a worst-case $2^{\Omega(\sqrt{n})}$ size lower bound for extended formulations of the max-knapsack polytope, which directly implies a similar result for the min-knapsack polytope. Other recent works include~\cite{extform4,bfps2012,CLRS13,Rothvoss14,LRS15,BFPS2015}.

\paragraph{Outline.} We prove our main result in Section~\ref{sec:protocol}, after giving preliminaries in Section~\ref{sec:prelim}. Instead of explicitly constructing our extended formulation, we provide a nonnegative factorization of the appropriate slack matrix. For this, we use the language of communication complexity --- we give an $O(\log^2 n + \log (1/\varepsilon))$-complexity two-party communication protocol with private randomness and nonnegative outputs whose expected output is the slack of a given feasible solution with respect to a given (weakened) knapsack cover inequality.

Next, in Section~\ref{sec:FCI}, we extend our communication protocol to the flow cover inequalities for the Single-Demand Facility Location problem, and show how to approximate the exponentially many flow cover inequalities using a smaller LP formulation.

Finally, in Section~\ref{sec:cutting_plane}, we show that although we do not know how to write down our extended formulation for min-knapsack in quasi-polynomial time, we can at least compute a $(2+\varepsilon)$-approximation of the optimum from the extended formulation in quasi-polynomial time, given any cost vector, \emph{without} relying on the ellipsoid algorithm. This is done via a new cutting-plane algorithm that might be of independent interest.

\section{Preliminaries.} \label{sec:prelim}
In this section, we introduce some key notions related to our problem.  We
review extended formulations and extension complexity of pairs of polyhedra in
Section~\ref{sec:sub:polytope}. Next, we define randomized communication
protocols with non-negative outputs that compute entries of matrices in
expectation.  Finally, in Section~\ref{sec:sub:circuits}, we review some
constructions of low-depth monotone circuits, and the Karchmer-Wigderson game
that relates circuit complexity and communication complexity. 

\subsection{Polyhedral Pairs, Extended Formulations and Slack Matrices.}
\label{sec:sub:polytope}

Let $P \subseteq \R^n$ be a polytope and $Q \subseteq \R^n$ be a polyhedron containing $P$. The complexity of the \emph{polyhedral pair} $(P,Q)$ can be measured by its extension complexity, which roughly measures how compactly we can represent a relaxation of $P$ contained in $Q$. The formal definition is as follows.

\begin{Definition}
Given a polyhedral pair $(P,Q)$ where $P \subseteq Q \subseteq \mathbb{R}^n$, we say that a system $E^{\leqslant} x + F^{\leqslant} y \leqslant g^{\leqslant}$, $E^{=}x + F^{=}y = g^{=}$ in $\R^{n+k}$ is an \emph{extended formulation} of $(P,Q)$ if the polyhedron $R := \{x \in \mathbb{R}^n \mid \exists y \in \R^k : E^{\leqslant} x + F^{\leqslant} y \leqslant g^{\leqslant}$, $E^{=}x + F^{=}y = g^{=}\}$ contains $P$ and is contained in $Q$. The \emph{size} of the extended formulation is the number of inequalities in the system. The \emph{extension complexity} of $(P,Q)$, denoted by $\xc(P,Q)$, is the minimum size of an extended formulation of $(P,Q)$. 
\end{Definition}

Although the case $P = Q$ is probably the most frequent, we will need polyhedral pairs here. In a seminal paper, Yannakakis~\cite{Yannakakis91} showed that one can study the extension complexity of a polytope $P$ through the non-negative rank of a matrix associated with $P$, namely, its slack matrix.

\begin{Definition}
Let $(P,Q)$ be a polyhedral pair with $P \subseteq Q \subseteq \R^n$. Let $P = \conv(\{v_1,\dots,v_p\})$ be an inner description of $P$ and $Q = \{x \in \R^n \mid Ax \geqslant b\}$ be an outer description of $Q$, where $A \in \R^{m \times n}$ and $b \in \R^m$. We now define the slack matrix $S$ of the pair $(P,Q)$ with respect to the given representations of $P$ and $Q$. The $i$th row of $S$ corresponds to the constraint $A_i x \geqslant b_i$, while the $j$th column of $S$ corresponds to the point $v_j$. The value $S_{i,j}$ measures how close the constraint $A_i x \geqslant b_i$ is to being tight for point $v_j$. More specifically, the \emph{slack matrix} $S \in \mathbb{R}^{m \times p}_{\geqslant 0}$ is defined as $S_{i,j} := A_i v_j - b_i$ for all $i \in [m], j \in [p]$. \end{Definition}

Note that the slack matrix is not unique as it depends on the choices of points $v_1, \ldots, v_p$ and linear description $Ax \geqslant b$. 

\begin{Definition}
Given a non-negative matrix $M \in \R_{\geqslant 0}^{m \times n}$, we say that a pair
of matrices $T, U$ is a \emph{rank-$r$ non-negative factorization} of $M$ if $T
\in \R_{\geqslant 0}^{m \times r}$, $U \in \R_{\geqslant 0}^{r \times n}$, and
$M = TU$. We define the \emph{non-negative rank} of $M$ as $\nnegrk(M) :=
\min\{ r: M \textrm{ has a rank-$r$ non-negative factorization}\}$. Notice that
a non-negative factorization of $M$ of rank at most $r$ is equivalent to
a decomposition of $M$ as a sum of at most $r$ non-negative rank-$1$ matrices.
\end{Definition}

Yannakakis~\cite{Yannakakis91} proved that for a polytope $P$ of dimension at least~$1$ and any of its slack matrices $S$, the extension complexity of $P$ is equal to the non-negative rank of $S$. Namely, $\xc(P) = \nnegrk(S)$. In particular, all the slack matrices of $P$ have the same nonnegative rank. 

This \emph{factorization theorem} can be extended to polyhedral pairs: we have $\xc(P,Q) \in \{\nnegrk(S),\nnegrk(S)-1\}$ whenever $S$ is a slack matrix of $(P,Q)$, see e.g.~\cite{bfps2012}.

\subsection{Randomized Communication Protocols.}
\label{sec:sub:commprotocol}

We now define a certain two-party communication problem and relate it to the non-negative rank discussed earlier,
following the framework in Faenza, Fiorini, Grappe and Tiwary~\cite{FaenzaFioriniGrappeTiwary11}.

\begin{Definition}\label{def:protocol}
  Let $S \in \R_{\geqslant 0}^{\mathcal{A} \times \mathcal{B}}$ be a non-negative matrix whose
  rows and columns are indexed by $\mathcal{A}$ and $\mathcal{B}$, respectively.
  Let $\Pi$ be a communication protocol with private randomness between two players Alice and Bob.
  Alice gets an input $a \in \mathcal{A}$ and Bob gets an input $b \in \mathcal{B}$. They exchange bits in a pre-specified way according to $\Pi$, and at the end either one of the players outputs some \emph{non-negative} number $\xi \in \R_{\geqslant 0}$. We say that $\Pi$ 
  \emph{computes $S$ in expectation} if for every $a$ and $b$, the expectation of the output $\xi$
  equals $S_{a,b}$.

  The \emph{communication complexity of a protocol $\Pi$} is the maximum of the number of bits exchanged \emph{between} Alice and Bob, over all pairs $(a,b) \in \mathcal{A} \times \mathcal{B}$ and the private randomness of the players. The size of the final output does not count towards the communication complexity of a protocol. The \emph{communication complexity of $S$}, denoted $\CCexp(S)$ is the minimum communication complexity of a randomized protocol $\Pi$ computing $S$ in expectation.
\end{Definition}

Faenza \emph{et al.}~\cite{FaenzaFioriniGrappeTiwary11} relate the non-negative rank of a non-negative matrix $S$, to the communication complexity $\CCexp(S)$. In particular, they prove that if $\nnegrk(S) \neq 0$, then $\CCexp(S) = \log_2 \nnegrk(S) + \Theta(1)$. Combining this with the factorization theorem, we get $\CCexp(S) = \log_2 \xc(P,Q) + \Theta(1)$ whenever $(P,Q)$ is a polyhedral pair with slack matrix $S$, provided that $\xc(P,Q) \neq 0$.

\subsection{Weighted Threshold Functions and Karchmer-Widgerson Game.}
\label{sec:sub:circuits}

An important part of our protocol depends on the communication complexity 
of (monotone) weighted threshold functions.
We start with the following result from~\cite{BW05,BW06} which gives low-depth
circuits for such functions. Another construction was given in~\cite{COS15}.
The circuits as stated in~\cite{BW05,BW06,COS15} have logarithmic depth, polynomial size
and unbounded fan-in, thus  it
is straightforward to convert them into circuits with fan-in $2$ with a logarithmic increase in depth.
Below we state the result for circuits of fan-in $2$ as will be used later.
Recall that a circuit is \emph{monotone} if it uses only AND and OR gates, but no NOT gates.

\begin{theorem}[\cite{BW05,BW06}]\label{thm:monotone-threshold-circuit}
  Let $w_1,\ldots,w_n \in \Z_{>0}$ be positive weights, and $T \in \Z_{\geqslant 0}$ be a threshold.
  Let $f: \{0,1\}^n \to \{0,1\}$ be the monotone function such that $f(x_1,\ldots,x_n)=1$ if and only if $\sum_{i=1}^{n} w_i x_i \geqslant T$.
  Then there is a depth-$O(\log^2 n)$ monotone circuit of fan-in $2$ that computes the function $f$.
\end{theorem}

The well-known \emph{Karchmer-Wigderson game} \cite{KM90} connects the depth of monotone circuits
and communication complexity. Given a monotone function $f:\{0,1\}^n \to \{0,1\}$,
the \emph{monotone Karchmer-Wigderson game} is the following: Alice receives $a \in f^{-1}(0)$, Bob receives $b \in f^{-1}(1)$, they communicate bits to each other, and the goal is to agree on a position $i \in \{1,\ldots,n\}$ such that $a_i = 0$ and $b_i = 1$.
Let $\CCKW(f)$ be the deterministic communication complexity of this game.

\begin{theorem}[\cite{KM90}]\label{thm:karchmer-wigderson}
  Let $f:\{0,1\}^n \to \{0,1\}$ be a monotone function, $\CCKW(f)$ be the deterministic communication complexity of the Karchmer-Wigerson game, and $\mathrm{depth}(f)$ be the minimum depth of a fan-in $2$
  monotone circuit that computes $f$. Then $\mathrm{depth}(f)=\CCKW(f)$.
\end{theorem}

Combining Theorems~\ref{thm:monotone-threshold-circuit} and \ref{thm:karchmer-wigderson}, we immediately get that $\CCKW(f) = O(\log^2 n)$ for every weighted threshold function $f$ on $n$ inputs.

\section{Small LP relaxation for Min-Knapsack.} \label{sec:protocol}

In this section, we show the existence of a $(1/\varepsilon)^{O(1)}n^{O(\log
n)}$-size LP relaxation of min-knapsack with integrality gap $2 + \varepsilon$,
proving Theorem~\ref{thm:main}. First, we give a high-level overview of the
construction in Section~\ref{sec:sub:overview}.  The actual protocol is
described and analyzed in Section~\ref{sec:sub:protocol}.

\subsection{Overview.}
\label{sec:sub:overview}

Consider the slack matrix $S$ that has one row for each knapsack cover inequality
and one column for each feasible solution of min-knapsack. More precisely, let 
$f : \{0,1\}^n \to \{0,1\}$ denote the weighted threshold function defined by 
the item sizes 
$\size_i$ ($i \in [n]$) and demand $D$ as in~(\ref{eq:wtf}). The rows and columns
of $S$ are indexed by $a \in f^{-1}(0)$ and $b \in f^{-1}(1)$ respectively.
The entries of $S$ are given by
$$
S_{a,b} := \sum_{i : a_i = 0} \rsize_i b_i - \rdemand\,,
$$
where as precedingly $U = U(a) := D - \sum_{i : a_i = 1} s_i$, and $\rsize_i = \rsize_i(a) = \min\{s_i, U\}$. Geometrically, $S$ is the slack matrix of the
polyhedral pair $(P,Q)$ in which $P$ is the min-knapsack polytope and $Q$ is 
the (unbounded) polyhedron defined by the knapsack cover inequalities.

Ideally, we would like to design a communication protocol for $S$, 
as those discussed in Section~\ref{sec:sub:commprotocol}, with low communication complexity.
This would imply a low-rank non-negative factorization of $S$. From the
factorization theorem of Section~\ref{sec:sub:polytope}, it would follow
that there exists a small-size extended formulation yielding a polyhedron 
$R$ containing the min-knapsack polytope $P$ and contained in the 
knapsack-cover relaxation $Q$. Hence, we would get a small-size LP 
relaxation for min-knapsack that implies the exponentially many knapsack 
cover inequalities, and thus have integrality gap at most $2$. 

However, due to the fact that the quantities involved can be 
exponential in $n$, making them too expensive to communicate 
directly, we have to settle for showing the existence of 
small-size extended formulation that \emph{approximately} 
implies the knapsack cover inequalities. Before discussing 
further these complications, we give an idealized version of
the protocol to help with the intuition. Assume for now that all item 
sizes and the demand are polynomial in $n$. Thus Alice and Bob 
can communicate them with $O(\log n)$ bits. 

The goal of the two players is to compute the slack 
$S_{a,b} = \sum_{i : a_i = 0} \rsize_i b_i - \rdemand$, 
when Alice is given an infeasible $a \in \{0,1\}^n$ and Bob 
is given a feasible $b \in \{0,1\}^n$. That is, after several 
rounds of communication, either one of them outputs some non-negative value 
$\xi$, such that the expectation of $\xi$ equals $S_{a,b}$. 

We define for a set of items $J \subseteq [n]$ 
the quantity 
$\size(J) := \sum_{j \in J} \size_j$, and
$\rsize(J) := \sum_{j \in J} \rsize_j$. Let $A$ and $B$ be the
subsets of $[n]$ corresponding to Alice's input $a$ and Bob's 
input $b$,  respectively. The slack we want to compute thus becomes 
$\rsize(B \smallsetminus A) - \rdemand$.

At the beginning, Alice computes the residual demand $\rdemand$ 
and sends it to Bob. Now observe that if there is some $i^* \in B 
\smallsetminus A$, such that $\size_{i^*} \geqslant \rdemand$, then 
we have $\rsize_{i^*} = \rdemand$, and we can easily write the slack as 
$\rsize(B \smallsetminus A \smallsetminus \{i^*\}) + (\rsize_{i^*}-U) 
= \rsize(B \smallsetminus A \smallsetminus \{i^*\})$ (similarly to the uncapacitated case discussed in the introduction).
Recall that we call an item $i$ large if $s_i \geqslant \rdemand$ 
and small otherwise. Let $\largeitems$ be the set of large 
items and $\smallitems$ be the set of small items. 

The rest of the protocol is divided into two cases as follows, depending on whether 
Alice and Bob can easily find a large item $i^* \in B \smallsetminus A$. 
To this end, Alice sends $\size(\largeitems \cap A)$ to Bob. Note 
that now Bob can compute $\size(\smallitems \cap A) = \demand
- \rdemand - \size(\largeitems \cap A)$. Bob computes the contribution 
of large items in $B$, that is, $\size(\largeitems \cap B)$. 

If $\size(\largeitems \cap B) > \size(\largeitems \cap A)$, 
then we are guaranteed that there is some $i^* \in \largeitems \cap (B \smallsetminus A)$.
Moreover, defining the threshold function 
\begin{equation}
\label{eq:truncation}
g(x) := 
\begin{cases}
1 &\text{if } \sum_{i \in \largeitems} \size_i x_i \geqslant \size(\largeitems \cap B),\\
0 &\text{otherwise},
\end{cases}
\end{equation}
then $g(a)=0$ and $g(b)=1$. Hence, Alice and Bob can find 
such an item with $O(\log^2 n)$ bits of communication, see 
Section~\ref{sec:sub:circuits}. With that, it is not hard to 
compute $\rsize(B \smallsetminus A \smallsetminus \{i^*\})$ 
with $O(\log n)$ bits of communication: Alice samples a 
uniformly random item $i$ and sends the index to Bob, Bob 
replies with $b_i$, Alice outputs $\rsize_i \cdot n$ if 
$b_i=1$, $i \ne i^*$ and $i \notin A$, and outputs $0$ 
otherwise. All her outputs are non-negative and their 
expectation is exactly the slack.

In the other case, $\size(\largeitems \cap B) \leqslant 
\size(\largeitems \cap A)$. Note that $\size(B) 
= \size(\largeitems \cap B) + \size(\smallitems \cap B) \geqslant \demand = \size(\largeitems \cap A) + \size(\smallitems \cap A)+\rdemand$, 
thus $\size(\smallitems \cap B)-\size(\smallitems \cap A) - \rdemand \geqslant \size(\largeitems \cap A) - \size(\largeitems \cap B) \geqslant 0$.
We now write the slack as 
\begin{align*}
 \rsize(B \smallsetminus A) - U =~ & \rsize(\largeitems \cap (B \smallsetminus A)) + \size(\smallitems \cap (B \smallsetminus A)) - \rdemand \\
  =~ &\rsize(\largeitems \cap (B \smallsetminus A)) + \size(\smallitems \cap B) - \size(\smallitems \cap (A \cap B)) - \rdemand \\
  =~ &\rsize(\largeitems \cap (B \smallsetminus A)) + \size(\smallitems \cap B) - \size(\smallitems \cap A)  + \size(\smallitems \cap (A \smallsetminus B)) - \rdemand \\
  =~& \rsize(\largeitems \cap (B \smallsetminus A)) + \size(\smallitems \cap (A \smallsetminus B))   + \left(\size(\smallitems \cap B) - \size(\smallitems \cap A) - \rdemand\right)
  \,.
\end{align*}
Alice and Bob can compute the first and the last term in expectation 
using a protocol similar to that in the previous case. The term in the 
middle can be computed by Bob with all the information he has at this 
stage. To conclude, in both cases, Alice and Bob can compute the exact 
slack $S_{a,b}$ with $O(\log^2 n)$ bits of communication.

\subsection{The Protocol.} \label{sec:sub:protocol}

The actual slack matrix $S^{\varepsilon}$ we work with is defined as
\begin{equation}
  \label{eq:relaxedKCIslack}
S^\varepsilon_{a,b} := \sum_{i : a_i = 0} \rsize_i b_i - \frac{2}{2+\varepsilon} \rdemand\,,
\end{equation}
where $\varepsilon > 0$ is any small constant, $a \in f^{-1}(0)$ and $b \in
f^{-1}(1)$. $S^{\varepsilon}$ is the slack matrix of the polyhedral pair $(P,
Q^{\varepsilon})$ where $P$ is the min-knapsack polytope and $Q^{\varepsilon}$
is the polyhedron defined by a slight weakening of the knapsack cover
inequalities  obtained by replacing the right hand side of~(\ref{eq:KC}) by
$\frac{2}{2+\varepsilon} \rdemand < \rdemand$. For every $x \in \R^n_{\geqslant
0}$ that satisfies all weakened knapsack cover inequalities, we have that
$\frac{2+\varepsilon}{2} x$ satisfies all original knapsack cover inequalities,
and thus $(2+\varepsilon) x$ dominates a convex combination of feasible
solutions. Therefore the integrality gap of the resulting LP relaxation
(obtained from a non-negative factorization of $S^{\varepsilon}$) is at
most $2+\varepsilon$. 
 
In order to refer to the ``derived'' weighted threshold functions $g$ as 
in~(\ref{eq:truncation}), we need a last bit of terminology. We say that
$g : \{0,1\}^n \to \{0,1\}$ is a \emph{truncation} of $f$ if there exists
$U, T \in \Z_{>0}$ with $T \leqslant D$ such that $g(x) = 1$ iff 
$\sum_{i = 1}^n w_i x_i \geqslant T$, where $w_i = s_i$ if 
$s_i \geqslant U$ and $w_i = 0$ otherwise. We are now ready to state  our main technical lemma. 

\begin{Lemma}\label{lem:mainlemma}
  For all constants $\varepsilon \in (0,1)$, item sizes $\size_i \in \Z_{> 0}$ 
  ($i \in [n]$), all smaller than $2^{\lceil n \log n \rceil}$ and 
  demand $\demand \in \Z_{> 0}$ with $\max \{s_i \mid i \in [n]\} 
  \leqslant D \leqslant \sum_{i=1}^n s_i$, such that all truncations 
  of the corresponding weighted threshold function admit 
  depth-$\depth$ monotone circuits of fan-in~$2$, there 
  is a $O(\log(1/\varepsilon)+\log n+\depth)$-complexity randomized 
  communication protocol with non-negative outputs that computes the slack 
  matrix $S^\varepsilon$ in expectation. Since we may always take 
  $\depth = O(\log^2 n)$, this gives a 
  $O(\log(1/\varepsilon)+\log^2 n)$-complexity protocol, unconditionally.
\end{Lemma}

Before giving the proof, let us remark that Theorem~\ref{thm:main} follows
directly from this lemma. Indeed,  the extra assumptions in the lemma are
without loss of generality: the fact that we may assume without loss of
generality that the item sizes $\size_i$ are positive integers that can be
written down with at most $\lceil n \log n \rceil$ bits, is due to a classic
result from~\cite{Muroga71}; and the fact that we may also assume that the demand
$D$ is a positive integer with $\max \{s_i \mid i \in [n]\} \leqslant
D \leqslant \sum_{i=1}^n s_i$ should be clear.

Moreover, Lemma~\ref{lem:mainlemma} implies that we can obtain a relaxation of polynomial size if all truncations of the weighted threshold function have monotone circuits of logarithmic depth. In particular, this is the case if all item sizes are polynomial in
$n$. In that case the threshold function (and its truncations) can  simply be written as
the majority function on $O(\sum_i s_i)$ input bits and, as such functions have
monotone circuits of fan-in~$2$ of logarithmic depth, i.e., depth $O(\log
\left( \sum_i s_i \right))$. Thus, using majority functions instead of threshold functions in our communication protocol, we get that for all $\varepsilon \in (0,1)$, $c>0$, item sizes $\size_1, \ldots, \size_n \in \{0, 1, \dots, n^c\}$  and demand $D \in \N$, there exists a size-$(1/\varepsilon)^{O(1)} n^{O(c)}$ extended formulation defining an LP relaxation of min-knapsack with integrality gap at most $2+\varepsilon$. However, it is important to note here that when $c$ is a constant (and hence the sizes $\size_1,\ldots,\size_n$ and the demand $D$ are polynomial in $n$), we can write down an exact polynomial size LP formulation of the min-knapsack problem\footnote{This can be done by casting the folklore \emph{Dynamic Programming} algorithm for the min-knapsack problem, as a minimium $s$-$t$ flow problem on a weighted graph $G$ with polynomial many vertices, and arguing that the well-known exact LP relaxation of the latter is also an exact LP relaxation of the former. The reader familiar with the dynamic programming algorithm should notice that the edges $G$ would only depend on the sizes of the items, whereas the weights on these edges would only depend on the costs of the items. Hence in the resulting LP formulation, the constraints depend only on the sizes, and the costs only appear in the objective function.}. 

We now proceed by proving our main technical result, i.e., Lemma~\ref{lem:mainlemma}.

\begin{proof}[Proof of Lemma~\ref{lem:mainlemma}]
Let $\alpha = \alpha(\varepsilon) := 2/(2 + \varepsilon)$ and 
$\delta > 0$ be such $(1-2\delta)/(1+\delta) = \alpha$. Thus 
$\delta = \varepsilon / (6+2\varepsilon) = \Theta(\varepsilon)$. 
As above, we denote by $a \in f^{-1}(0)$ the input of Alice and
$b \in f^{-1}(1)$ that of Bob, and let $A$ and $B$ denote the 
corresponding subsets of $[n]$.

First, Alice tells Bob the identity of the set of large items $\largeitems 
= \{i \in [n] \mid \size_i \geqslant U\}$ and its complement, the 
set of small items $\smallitems$. This costs $O(\log n)$ bits of 
communication. For instance, Alice can simply send the index of 
a smallest large item to Bob, or inform Bob that $\largeitems$ is 
empty. Recall that
$$
\rdemand = D - \size(A) = D - \size(\largeitems \cap A) - \size(\smallitems \cap A)\,.
$$

Then, Alice sends Bob the unique nonnegative integer $k$ such that 
$(1 + \delta)^k \leqslant \rdemand < (1+\delta)^{k+1}$. This sets the 
scale at which the protocol is operating. 
Since $\rdemand \leqslant n \cdot 2^{\lceil n \log n \rceil} \leqslant 2^{n^2}$, 
we have $(1+\delta)^k \leqslant 2^{n^2}$. 
This implies that $k = O((1/\varepsilon) n^2)$, thus $k$ can be sent to Bob with 
$\log (1/\varepsilon) + 2\log n + O(1) = O(\log(1/\varepsilon) + \log n)$ bits.
Let $\apxrdemand := (1 + \delta)^k$. 

To efficiently communicate an approximate value of $\size(\largeitems \cap A)$,
 Alice sends the unique nonnegative integer $\ell$ such that 
$$
(1 + \ell \delta) \apxrdemand < \demand - \size(\largeitems \cap A) \leqslant (1 + \ell \delta) \apxrdemand + \delta \apxrdemand.
$$
Since small items have size at most $\rdemand$ and we have at most $n$ of them, 
we have $\size(\smallitems \cap A) \leqslant \rdemand n$. Hence, $\demand - 
\size(\largeitems \cap A) = \rdemand + \size(\smallitems \cap A) \leqslant (n+1) 
\rdemand \leqslant (n+1) (1+\delta) \apxrdemand$. Since 
$(1+\ell\delta)\apxrdemand < (n+1)(1+\delta)\apxrdemand$, we have 
$\ell = O((1/\varepsilon)n)$. This means that Alice can communicate 
$\ell$ to Bob with only $O(\log(1/\varepsilon)+\log n)$ bits. Let
$\tilde{\Delta} = \tilde{\Delta}(\delta) := (1 + \ell \delta) \apxrdemand$. This is Bob's 
strict under-approximation of $\demand - \size(\largeitems \cap A)$, so that 
$\demand - \tilde{\Delta}$ is a strict over-approximation of $\size(\largeitems \cap A)$.

Bob checks if $s(\largeitems \cap B) \geqslant \demand - \tilde{\Delta}$. If this is the case, then the weighted threshold function $g$ such that $g(x) = 1$ iff $\sum_{i \in \largeitems}\size_i x_i \geqslant \demand - \tilde{\Delta}$ separates $a$ from $b$ in the sense that $g(a) = 0$ and $g(b) = 1$. Since $g$ is a truncation of $f$, Alice and Bob can exchange $\depth$ bits to find an index $i^* \in \largeitems$ such that $a_{i^*} = 0$ and $b_{i^*} = 1$. 

We can rewrite the slack $S^{\varepsilon}_{a,b} = \rsize(B \smallsetminus A) - \alpha \rdemand$ as
\begin{align}
 \label{eq:slackbig}
   \rsize(B \smallsetminus A \smallsetminus \{i^*\}) + \rsize_{i^*} - \alpha \rdemand 
   =~  \rsize(B \smallsetminus A \smallsetminus \{i^*\}) + (\rdemand - \alpha \rdemand)
 =~  \sum_{i : a_i = 0,\ i \neq i^*} \rsize_i b_i + (\rdemand - \alpha \rdemand)
\,.
\end{align}
With the knowledge of $i^*$, Alice and Bob can compute the slack as follows:
\begin{enumerate}
  \item Alice samples a uniformly random number $i \in [n]$. If $i \notin A$, continue
    to the next step, otherwise Alice outputs $0$ and terminates the communication. 
  \item If $i=i^*$, Alice outputs $n \cdot (U-\alpha U)$ and terminates the communication,
    otherwise continue.
  \item Alice sends $i$ to Bob using $\lceil \log n \rceil$ bits of communication,
    and Bob sends $b_i$ back to Alice.
  \item Alice outputs $n \cdot s_i' b_i$.
\end{enumerate}
The above communication costs $O(\log n)$ bits, all outputs are non-negative
and can be computed with the information available to each player,
and by linearity of expectation, the expected output is exactly the slack~(\ref{eq:slackbig}).
Together with the $O(\log(1/\varepsilon) + \log n + t)$ bits communicated 
previously, we conclude that in this case there is a protocol that computes
the slack in expectation with $O(\log(1/\varepsilon) + \log n + t)$ bits of
communication.

In the other case, we have $s(\largeitems \cap B) < \demand - \tilde{\Delta}$. Because $b \in \{0,1\}^n$ is feasible, we get
$$
s(B) \geqslant \demand
\iff
\underbrace{s(\largeitems \cap B)}_{< \demand - \tilde{\Delta}} + s(\smallitems \cap B) \geqslant D\,,
$$
therefore we can bound $s(\smallitems \cap B)$ as
\begin{align}
   \label{eq:lbsmallcontrib}
s(\smallitems \cap B)  > \tilde{\Delta}  \geqslant ~ \demand - \size(\largeitems \cap A) - \delta \apxrdemand  = ~ \size(\smallitems \cap A) + \rdemand - \delta \apxrdemand 
\geqslant ~ \apxsmallcontrib + (1-\delta) \apxrdemand\,,
\end{align}
where $\apxsmallcontrib$ is the unique integer multiple of $\delta \apxrdemand$ such that
\begin{equation}
  \label{eq:lbsmall1}
\apxsmallcontrib \leqslant \size(\smallitems \cap A) < \apxsmallcontrib + \delta \apxrdemand\,.
\end{equation}

Since $\apxsmallcontrib \leqslant \size(\smallitems \cap A) \leqslant \rdemand n \leqslant (1+\delta) \apxrdemand n$, Alice can communicate $\apxsmallcontrib$ to Bob with $O(\log (1/\varepsilon) + \log n)$ bits.

This implies
\begin{align*}
 \size(\smallitems \cap (B \smallsetminus A)) 
= ~  \size(\smallitems \cap B) - \size(\smallitems \cap (A \cap B)) 
> ~ \apxsmallcontrib + (1-\delta) \apxrdemand - \size(\smallitems \cap (A \cap B))\,.
\end{align*}

Recall that by definition of $\apxrdemand$, we have $(1+\delta)\apxrdemand > \rdemand$, therefore
\begin{equation}
  \label{eq:lbsmall2}
  (1-2\delta)\apxrdemand - \alpha \rdemand > (1-2\delta)\apxrdemand - \alpha (1+\delta) \apxrdemand=0\,.
\end{equation}

We now rewrite the slack as
\begin{align*}
 \rsize(B \smallsetminus A) - \alpha \rdemand 
= ~& \underbrace{\rsize(\largeitems \cap (B \smallsetminus A))}_{=~ \sum_{i \in \largeitems \smallsetminus A} \rsize_i b_i} + \underbrace{\size(\smallitems \cap B) - \apxsmallcontrib - (1-\delta) \apxrdemand}_{\textrm{non-negative by~(\ref{eq:lbsmallcontrib})}} + \underbrace{\size(\smallitems \cap (A \smallsetminus B))}_{\sum_{i \in \smallitems \cap A} \size_i (1-b_i)} \\
& + \underbrace{\apxsmallcontrib - \size(\smallitems \cap A) + (1-\delta) \apxrdemand - \alpha \rdemand}_{\textrm{non-negative by~(\ref{eq:lbsmall1}) and~(\ref{eq:lbsmall2})}}\,.
\end{align*}
Similar to the previous case, we design a protocol to compute the slack as follows:
\begin{enumerate}
  \item Alice samples a uniformly random number $i \in [n+2]$. 
    If $i=n+2$, Alice outputs the normalized value of the last term, i.e., 
    $(n+2) \cdot (\apxsmallcontrib - \size(\smallitems \cap A) + (1-\delta) \apxrdemand - \alpha \rdemand)$, 
    and terminates the communication. Otherwise, she sends $i$ to Bob using $O(\log n)$ bits.
  \item If $i=n+1$, Bob outputs $(n+2) \cdot (\size(\smallitems \cap B) - \apxsmallcontrib - (1-\delta) \apxrdemand)$, and ends the communication. Otherwise, he replies to Alice with $b_i$.
  \item If $i \in \largeitems \smallsetminus A$, Alice outputs $(n+2) \cdot \rsize_i b_i$;
    if $i \in \smallitems \cap A$, she outputs $(n+2) \cdot \size_i (1-b_i)$; otherwise she outputs $0$.
\end{enumerate}
We can verify that the outputs of both players can be computed with information available
to them, and that the outputs are non-negative due to Equation~(\ref{eq:lbsmallcontrib}),~(\ref{eq:lbsmall1}) and~(\ref{eq:lbsmall2}), and the definition of the variables.
\end{proof}

\section{Flow-cover inequalities.}
\label{sec:FCI}
A variant of the knapsack  cover inequalities, known as the \emph{flow cover inequalities}, was also used to strengthen LPs for many problems such as the Fixed Charge Network Flow problem~\cite{CFLP00} and the Single-Demand Facility Location problem~\cite{CS08}. 
In this section, we describe the application of flow cover inequalities to the Single-Demand Facility Location problem as used in~\cite{CS08}, and then give an $O(\log^2 n)$-bit two-party communication protocol that computes a weakened version of these inequalities.

In the Single-Demand Facility Location problem, we are given a set $F$ of $n$ facilities, such that each facility $i \in F$ has a capacity $s_i$, an opening cost $f_i$, and a per-unit cost $c_i$ to serve the demand. The goal is to serve the demand $D$ by opening a subset $S\subseteq F$ of facilities such that the combined cost of opening these facilities and serving the demand is minimized. The authors of~\cite{CS08} cast this problem as an Integer Program, and showed that its natural LP relaxation has an unbounded integrality gap. To reduce this gap, they strengthened the relaxation by adding the so-called flow cover inequalities that we define shortly (See Section 3 in~\cite{CS08} for a more elaborate discussion). 

A feasible solution $(x,y)$ with $y \in \{0,1\}^n$ and $x\in [0,1]^n$ for the Single-Demand Facility Location LP can be thought of as follows: 
for each $i\in F$, $y_i \in \{0,1\}$ indicates if the $i$-th facility is open, and $x_i \in [0,1]$ indicates the fraction of the demand $D$ being served by the $i$-th facility. A feasible solution $(x,y)$ must then satisfy that \begin{enumerate}
\item The demand is \emph{met}, i.e., $\sum_i x_i  =1$.
\item No facility is supplying more than its capacity, i.e., $0 \leqslant x_i D \leqslant y_i s_i $ for all $i\in F$. 
\end{enumerate}
For a subset $J \subseteq F$ of facilities and a feasible solution $(x,y)$, we denote by $B = \{i \in F: y_i = 1\} \subseteq [F]$ the set of \emph{open} facilities according to $y$, and we define the quantity $\xsize(J)$ to be the overall demand served by the facilities in $J$, i.e., $\xsize(J) = \sum_{i \in J} x_i D$.\footnote{Note that since we are assuming that $(x,y)$ is feasible, we get that $x(J)=x(J \cap B)$.} 
We also define the quantities $\size(\cdot)$ and $\rsize(\cdot)$ as in Section~\ref{sec:sub:overview}. 

Carnes and Shmoys~\cite{CS08} showed that adding the flow cover inequalities (FCI) reduces the integrality gap of the natural LP relaxation down to 2. These inequalities are defined as follows:
for any \emph{infeasible} set $A \subseteq F$ (i.e.,  $A\subseteq F$ such that $\size(A) < D$), 
and for all partitions of $F \setminus A = F_1 \sqcup F_2$,
the following inequality holds for all feasible solutions $(x,y)$: 
\begin{align}
\label{eq:fci1}
\rsize(F_1 \cap B) + \xsize (F_2 \cap B) \geqslant U\,, \tag{FCI}
\end{align} 
where $U = D - \size(A)$ is the residual demand and $\rsize_i = \min\{\size_i, U\}$.
For brevity, we refer to an infeasible set $A$ along with some partition $F_1 \sqcup F_2=F \setminus A$ as an \emph{infeasible tuple} $(A,F_1,F_2)$. Note that for $F_2 = \emptyset$, the flow-cover inequalities are the same as the knapsack cover inequalities.

Similar to the knapsack cover inequalities, the goal is to compute the slack of a \emph{relaxed} version of~(\ref{eq:fci1}) in expectation for any feasible solution $(x,y)$ and any infeasible tuple $(A,F_1,F_2)$. Namely, for any $\varepsilon \in (0,1)$, let $\alpha=2/(2+\varepsilon)$, then our goal is to design an $O(\log ^2 n+\log(1/\varepsilon))$-complexity two-party communication protocol with private randomness and nonnegative outputs whose expected output equals $\rsize(F_1 \cap B) + \xsize (F_2 \cap B) - \alpha U$.
 That is, we want to compute the slack with respect to a given (weakened) flow-cover inequality 
$\rsize(F_1 \cap B) + \xsize (F_2 \cap B) \geqslant \alpha U$, where the RHS of~(\ref{eq:fci1}) is replaced by $\alpha U$.
This implies the existence of an LP of size $(1/\varepsilon)^{O(1)} n^{O(\log n)}$ with an integrality at most $2+\varepsilon$ for the Single-Demand Facility Location problem.

In Section~\ref{sec:fci:notation}, we set up the notation and define a class of feasible solutions with a certain special structure which we refer to as \emph{canonical} feasible solutions. We design the promised communication protocol restricted to canonical solutions in Section~\ref{sec:fci:protocolcanonical}, and extend it  to arbitrary feasible solutions in Section~\ref{sec:fci:protocol}.

\subsection{Preliminaries.}
\label{sec:fci:notation}
Let $(x,y)$ be a feasible solution for the flow-cover problem with demand $D$, and let $B = \{i \in F: y_i = 1\}$ denote the support of $y$. In this terminology, $B$ only indicates which facilities are open, but it does not capture the \emph{relative} demand being served through each of them.
However this distinction will be essential for designing the protocol, hence we partition $B$ into three disjoint sets 
$B = \widetilde{F}_1 \sqcup \widetilde{F}_2 \sqcup \widetilde{F}_3$ as follows:
\begin{align*}
  \widetilde{F}_1 &= \{i \in B: x_i D = s_i y_i\}\,, &&
  \widetilde{F}_2  &= \{i \in B: 0<x_i D < s_i y_i\}\,, &&
  \widetilde{F}_3  &= \{i \in B: x_i D= 0\} \,.
\end{align*}
We first focus on feasible solutions $(x,y)$ that exhibit a certain structure, and then generalize to arbitrary solutions. Namely, we restrict our attention here and in Section~\ref{sec:fci:protocolcanonical} to \emph{canonical} feasible solutions defined as follows:
\begin{Definition}
\label{def:canonical}
A feasible solution $(x,y)$ with associated sets $\widetilde{F}_1,\widetilde{F}_2,\widetilde{F}_3$ is \emph{canonical} if $\widetilde{F}_2$ contains at most one facility, i.e., $|\widetilde{F}_2| \leqslant 1$. 
In other words, in a canonical feasible solution, there is at most one facility $j$ 
that supplies a non-zero demand $x_jD>0$ which is \emph{not} equal to its full capacity $s_j$.
\end{Definition}
Recall that we are interested in computing 
\begin{align}
\label{eq:FCI}
\rsize(F_1 \cap B) + \xsize(F_2 \cap B) - \alpha U 
\end{align} 
in expectation, which can be expanded as follows:
\begin{align}
 \label{eq:fciexpanded}
 \rsize(F_1 \cap \widetilde{F}_1) 
  + \rsize(F_1 \cap \widetilde{F}_2)
  + \rsize(F_1 \cap \widetilde{F}_3)
+ \xsize(F_2 \cap \widetilde{F}_1)
  + x(F_2 \cap \widetilde{F}_2) 
  +x(F_2 \cap \widetilde{F}_3) - \alpha U \,.
\end{align}
We get from the definition of the set $\widetilde{F}_3$ that the second to last term in the above equation is $0$ when restricted to canonical feasible solutions. In fact, one can completely get rid of the overall contribution of $\widetilde{F}_3$ in the above equation, since intuitively, \emph{closing} down the facilities in $\widetilde{F}_3$ should not alter the feasibility of the solution, and hence Equation~(\ref{eq:fciexpanded}) should still be positive even without accounting for the contribution of  $\rsize(F_1 \cap \widetilde{F}_3)$. In the communication protocol setting, this intuition translates to designing a protocol that only deals with canonical feasible solutions restricted to $\widetilde{F}_3=\emptyset$.
 
To see that this is without loss of generality, consider a canonical feasible solution $(x,y)$ such that $\widetilde{F}_3$ is \emph{not} empty, and let $(x,\bar{y})$ be the projection of $(x,y)$ on $\widetilde{F}_1\cup \widetilde{F}_2$ --- 
that is, for all $i \in B \setminus \widetilde{F}_3$, set $\bar{y}_i=y_i$, and for all $i \in \widetilde{F}_3$, set $\bar{y}_i=0$. It follows that $(x,\bar{y})$ is also a canonical feasible solution, as the items whose support is $\widetilde{F}_3$ do not contribute to the feasibility of the solution, and the cardinality of $\widetilde{F}_2$ does not change. Thus, for any infeasible tuple $(A,F_1,F_2)$, Equation~(\ref{eq:fciexpanded}) applied to $(x,\bar{y})$ can be written as 
\begin{align}
   \label{eq:fciexpanded2}
   \rsize(F_1 \cap \widetilde{F}_1) 
  + \rsize(F_1 \cap \widetilde{F}_2)   + \xsize(F_2 \cap \widetilde{F}_1) 
  + \xsize(F_2 \cap \widetilde{F}_2) 
   - \alpha U \,,
\end{align}
which is also non-negative, as it is the slack of $(x,\bar{y})$ and $(A,F_1,F_2)$. 
Therefore, for any feasible solution $(x,y)$, the slack as given by Equation~(\ref{eq:fciexpanded}) can be viewed
as the summation of Equation~(\ref{eq:fciexpanded2}) and the non-negative term $\rsize(F_1 \cap \widetilde{F}_3)$.
The latter is easy to compute with a small communication protocol\footnote{To compute $\rsize(F_1 \cap \widetilde{F}_3)$, Bob samples an index $i \in [n]$. If $i \notin \widetilde{F}_3$, he outputs 0 and terminates the protocol, otherwise he sends $i$ to Alice. If $i \in F_1$, Alice outputs $n\cdot \rsize(i)$, otherwise, she outputs 0.}, thus if Alice and Bob can devise a communication protocol $\Pi$ that computes~(\ref{eq:fciexpanded2}) in expectation, they can then easily compute~(\ref{eq:fciexpanded}) in expectation. For example, Alice can generate a uniformly random bit $b\in \{0,1\}$, and \begin{itemize}
\item if $b=0$, then Alice and Bob run the protocol that computes $\rsize(F_1 \cap \widetilde{F}_3)$, and return \emph{twice} its output.
\item if $b=1$, then Alice and Bob run the protocol $\Pi$ that computes~(\ref{eq:fciexpanded2}), and return \emph{twice} its output.
\end{itemize}

Moreover, since $|\widetilde{F}_2| \leqslant 1$, and using the fact that $x_i D = s_i y_i$ for $i\in \widetilde{F}_1 $, we can further simplify Equation~(\ref{eq:fciexpanded2}) as follows:
\begin{align}
\label{eq:fcisimplified}
\rsize(F_1 \cap  \widetilde{F}_1) 
+ \size(F_2 \cap \widetilde{F}_1)  
+ \gamma(x,y,A,F_1,F_2) - \alpha U \,,
\end{align}
where the function $ \gamma \coloneqq  \gamma(x,y,A,F_1,F_2)$ is defined as
\begin{align}
\label{eq:gamma}
 \gamma =
  \begin{cases}
    \rsize_j y_j  & \quad \text{if } \widetilde{F}_2 = \{j\} \subseteq  F_1\\
    x_j D  & \quad \text{if } \widetilde{F}_2 = \{j\} \subseteq F_2\\
        0       & \quad \text{if }   \widetilde{F}_2 = \{ j\} \subseteq A\text{, or } \widetilde{F}_2 = \emptyset  \,.\\
  \end{cases}
\end{align}
For simplicity of notation, we drop the parameters from $\gamma(x,y,A,F_1,F_2)$ when its is clear from the context.

\subsection{Randomized Protocol for Canonical Feasible Solutions.}
\label{sec:fci:protocolcanonical}
In what follows, we define a randomized communication protocol where Alice gets an infeasible tuple $(A,F_1,F_2)$, and Bob gets a \emph{canonical} feasible solution $(x,y)$ with $\widetilde{F}_3 = \emptyset$, and the goal is to compute the value of~(\ref{eq:fcisimplified}) in expectation. 
 
For a fixed $\varepsilon>0$, we define $\alpha:=\alpha(\varepsilon)=2/(2+\varepsilon)$,
$\delta:=\delta(\varepsilon)=\varepsilon/(6+2\varepsilon)$ as in the min-knapsack case.
Similar to the protocol for the knapsack cover inequalities, 
Alice sends Bob $O(\log n)$ bits at the beginning so that Bob knows $\largeitems, \smallitems,\apxrdemand,\apxsmallcontrib$ and $\widetilde{\Delta}$. 
Recall that $\largeitems$ is the set of large items (i.e., $i \in F$ such that $s(i) \geqslant U$), $\smallitems$ is the set of small items,
$\apxrdemand$ is an under-approximation of the residual demand $\rdemand$, 
$D-\widetilde{\Delta}$ is an over-approximation of $\size(\largeitems \cap A)$ and
$\apxsmallcontrib$ is an under-approximation of $\size(\smallitems \cap A)$.
Moreover, knowing his input $(x,y)$, Bob can construct the sets $\widetilde{F}_1$ and $\widetilde{F}_2$. 
Thus, by exchanging an additional $O(\log n)$ bits, Alice and Bob can both figure out which condition is satisfied for Equation~(\ref{eq:gamma}).

To compute the value of~(\ref{eq:fcisimplified}) in expectation, we distinguish between the following cases: 
\begin{description}
  \item[Case 1:] Either $\widetilde{F}_2 = \emptyset$, or $\widetilde{F}_2=\{j\}$ and $j \in A \cup F_1$.
    In this case, we have that the value $\gamma $ is either $0$ or $\rsize_j y_j$. 
    Bob now checks if 
    \begin{align}
      \label{eq:case1threshold} 
        s(\largeitems \cap (\widetilde{F}_1 \cup \widetilde{F}_2)) \geqslant D - \widetilde{\Delta}\,.
    \end{align}
      {\bf Equation~(\ref{eq:case1threshold}) holds:} 
        In the same way as in the min-knapsack protocol, 
        Alice and Bob exchange $O(\log^2 n)$ bits to identify an index $i^* \in \largeitems$ 
        such that $i^* \in ((\widetilde{F}_1 \cup \widetilde{F}_2) \setminus A)$. More precisely, this index $i^*$  belongs to one of the following three sets: either $i^* \in F_1 \cap \widetilde{F}_1$, or $i^* \in F_2 \cap \widetilde{F}_1$, or $i^*=j$ and $\widetilde{F}_2=\{j\}$. Alice and Bob can thus exchange $O(1)$ more bits to figure out the condition that $i^*$ satisfies. In what follows, we design an $O(\log n)$-communication protocol to handle each of these cases.

        If $i^* \in F_2 \cap \widetilde{F}_1$, then Equation~(\ref{eq:fcisimplified}) can be rewritten as 
        \begin{align}
        \label{eq:fci:case1}
                  \rsize( F_1 \cap \widetilde{F}_1) 
          + \size((F_2 \cap \widetilde{F}_1)\setminus\{i^*\})
          + {\gamma} + \left(s_{i^*} - \alpha U\right)\,.
        \end{align}
        One can see that each of the above four terms is non-negative, and similar to the min-knapsack protocol, Alice and Bob can exchange $O(\log n)$ bits and compute the value of~(\ref{eq:fci:case1}) as follows: \begin{enumerate}
          \item Bob sends Alice the bit $y_j$ and the index $j$ using $\lceil \log (n) \rceil + 1$ bits if and only if $\widetilde{F}_2 = \{j \}$, and he sends 0 if $\widetilde{F}_2 = \emptyset$.
        \item Alice samples a uniformly random index $i \in [n+1]$. 
          If $i = n+1$, Alice uses the knowledge of $\widetilde{F}_2$ (and thus $\gamma$) 
          to compute the normalized value of the last terms, 
          that is, she outputs $(n+1)\cdot \left(\gamma + s_{i^*} - \alpha U\right)$, 
          and terminates the communication. Otherwise, she sends $i$ to Bob using $\lceil \log(n) \rceil$ bits.
        \item If $i \in \widetilde{F}_1$, Bob sends $y_i$ to Alice; otherwise, Bob outputs 0 and terminates the communication.
        \item If $i\in F_1$, Alice outputs $(n+1) \cdot \rsize_i y_i$; if $i \in F_2\setminus \{i^*\}$, she outputs $(n+1)\cdot s_i y_i$; otherwise she outputs 0. 
        \end{enumerate}
        The above communication costs $O(\log n)$ bits, all outputs are non-negative
and can be computed with the information available to each player,
and by linearity of expectation, the expected output is exactly the slack~(\ref{eq:fcisimplified}) when $i^* \in F_2 \cap \widetilde{F}_1$.

        The case where $i^* \in F_1 \cap \widetilde{F}_1$ is handled similarly.
        
        In the remaining case, we have $\widetilde{F}_2=\{j\}$ and $i^*=j \in F_1 \cap \largeitems$, and hence 
        $\gamma=\rsize_j y_j > \alpha U$. This can be handled by changing the second step of the protocol described earlier in such a way that Alice outputs $(n+1)\cdot(\rsize_j - \alpha U)$ if $i =n+1$.
      
      {\bf Equation~(\ref{eq:case1threshold}) does not hold: }
        Recall that since $(x,y)$ is a feasible solution (and $\widetilde{F}_3=\emptyset$), we have
        \begin{align*}
          D ~\le~& \xsize(\widetilde{F}_1) + \xsize(\widetilde{F}_2) \\
          ~ = ~ & \xsize(\smallitems \cap \widetilde{F}_1) + \xsize(\smallitems \cap \widetilde{F}_2)
          + \xsize(\largeitems \cap \widetilde{F}_1) + \xsize(\largeitems \cap \widetilde{F}_2) \\
          ~\le~ & \xsize(\smallitems \cap \widetilde{F}_1) + \xsize(\smallitems \cap \widetilde{F}_2) 
         + \size(\largeitems \cap \widetilde{F}_1) + \size(\largeitems \cap \widetilde{F}_2) \\
          ~= ~ & \xsize(\smallitems \cap \widetilde{F}_1) + \xsize(\smallitems \cap \widetilde{F}_2) 
           + \size(\largeitems \cap (\widetilde{F}_1 \cup \widetilde{F}_2))\,.
        \end{align*}
        By the assumption that Equation~(\ref{eq:case1threshold}) does not hold, together
        with the argument in Equation~(\ref{eq:lbsmallcontrib}),
        we conclude that
        \begin{equation}
          \label{eq:cas1non}
          \xsize(\smallitems \cap \widetilde{F}_1) + \xsize(\smallitems \cap \widetilde{F}_2) 
          > \widetilde{\Delta} \geqslant \apxsmallcontrib + (1-\delta) \apxrdemand\,.
        \end{equation}
        Note that since $|\widetilde{F}_2| \leqslant 1$, we get that
        \begin{align*}
          \xsize(\smallitems \cap \widetilde{F}_2)=
          \begin{cases}
            0       & \quad \text{if } \widetilde{F}_2 = \emptyset \\
            0       & \quad \text{if } \widetilde{F}_2 = \{j\} \subseteq \largeitems\\
            x_j D  & \quad \text{if } \widetilde{F}_2 = \{j\} \subseteq \smallitems\,.\\
          \end{cases}
        \end{align*}
       We also have that 
        $\xsize(\smallitems \cap \widetilde{F}_1) = \size(\smallitems \cap \widetilde{F}_1)$  by the definition of $\widetilde{F}_1$.
        Together this gives that the summation $\size(\smallitems \cap \widetilde{F}_1)  + \xsize(\smallitems \cap \widetilde{F}_2)$ is lower bounded by $\apxsmallcontrib + (1-\delta) \apxrdemand$.
        We rewrite~(\ref{eq:fcisimplified}) as 
        \begin{align}
          \label{eq:case1expanded}
                  &\rsize(F_1 \cap \widetilde{F}_1) 
          + \size(F_2 \cap \widetilde{F}_1)  + \gamma - \alpha U \\
		\nonumber
          = ~& \rsize(\largeitems \cap F_1 \cap \widetilde{F}_1)
          + \size(\largeitems \cap F_2 \cap \widetilde{F}_1)  
            + \size(\smallitems \cap (\widetilde{F}_1 \setminus A))   
           + \gamma - \alpha U\\
           \nonumber
             =~ & \rsize(\largeitems \cap F_1 \cap \widetilde{F}_1)
          + \size(\largeitems \cap F_2 \cap \widetilde{F}_1) 
          + \size(\smallitems \cap (A \setminus B))    
          + \size(\smallitems \cap A \cap \widetilde{F}_2 )  \\ \nonumber&+ \size( \smallitems \cap \widetilde{F}_1)    - \size(\smallitems \cap A)
         + \gamma - \alpha U \,.
        \end{align}
        The non-negativity of the first three terms is straightforward, and Alice and Bob
        can compute them by exchanging $O(\log n)$ bits\footnote{For instance, to compute 
          $\rsize(\largeitems \cap F_1 \cap \widetilde{F}_1)$, Alice samples uniformly $i \in [n]$ and sends
          it to Bob, Bob responds with $b=1$ if $i \in \widetilde{F}_1$ and $b=0$ otherwise. Alice then outputs
        $n \cdot \rsize_i$ if $i \in \largeitems \cap F_1$ and $b=1$, and $0$ otherwise. The protocols for the
      second and the third term are very similar.}. 
        By adding and subtracting $(\apxsmallcontrib + (1-\delta) \apxrdemand - \xsize(\smallitems \cap \widetilde{F}_2))$ to the remaining terms in~(\ref{eq:case1expanded}), we can rearrange the terms and rewrite the rest as the sum of the following three non-negative terms that we can easily compute: 
        \begin{align}
       &\left(
      \size(\smallitems \cap \widetilde{F}_1)
       -\apxsmallcontrib - (1-\delta) \apxrdemand + \xsize(\smallitems \cap \widetilde{F}_2)\right)  +
       \left(\apxsmallcontrib  + (1-\delta) \apxrdemand- \alpha U - \size(\smallitems \cap A)\right) \nonumber \\ \label{eq:secondpartofFCI} & + 
       \left( 
\size(\smallitems \cap A \cap \widetilde{F}_2)
       + \gamma - \xsize(\smallitems \cap \widetilde{F}_2)\right)\,.
     \end{align} 

        The non-negativity of the first part follows from~(\ref{eq:cas1non}), and Bob has all the information to compute it on his own.
        The non-negativity of the second 
        part follows from our definition of $\apxsmallcontrib$ and $\apxrdemand$, 
        and their relation to $\delta$ and $\alpha$. Moreover, Alice has all the information to compute
        this part.

        To see that the third part (i.e., $\size(\smallitems \cap A \cap \widetilde{F}_2)
       + \gamma - \xsize(\smallitems \cap \widetilde{F}_2)$) is also non-negative and can easily be computed by one of the players, note that: \begin{enumerate}
        \item If $\xsize(\smallitems \cap \widetilde{F}_2)=0$, then clearly it is non-negative. 
        In this case, Bob communicates the set $\widetilde{F}_2$ to Alice using $O(\log n)$ bits so that she knows whether $\widetilde{F}_2=\emptyset$, or the item $j$ if $\widetilde{F}_2 = \{j\}$ and $j \in \largeitems$. Once $\widetilde{F}_2$ is known to Alice, she can compute both $\size(\smallitems \cap A \cap \widetilde{F}_2)$ and $\gamma$ (recall that $\gamma$ would be either 0 or $\rsize_j y_j = U$).
        \item If $\xsize(\smallitems \cap \widetilde{F}_2) = x_j D \ne 0$, then we have that $\widetilde{F}_2 = \{j\}$ and $j \in \smallitems$.  
        From our assumption of Case 1, we also have that $j \in A \cup F_1$. Since $A$ and $F_1$ are two disjoint sets, we get that: \begin{enumerate}
        \item If $j\in A$, then \begin{align*}
        \underbrace{\size(\smallitems \cap A \cap \widetilde{F}_2)}_{s_j y_j}+ \underbrace{\gamma}_{0} - x_j D 
         = ~s_j y_j - x_j D  \geqslant 0\,.
        \end{align*}
        \item If $j \in F_1$, then\begin{align*}
        \underbrace{\size(\smallitems \cap A \cap \widetilde{F}_2)}_{0}+ \underbrace{\gamma}_{s_jy_j} - x_j D =~ s_j y_j - x_j D  \geqslant 0\,.
        \end{align*} 
        \end{enumerate}
        Thus  it is also non-negative, and Bob can compute it on his own in this case.
       \end{enumerate}
This concludes the communication problem in the case where either $\widetilde{F}_2 = \emptyset$, or $\widetilde{F}_2=\{j\}$ where $j \in A \cup F_1$.
  \item[Case 2:] $\widetilde{F}_2=\{j\}$ and $j \in F_2$. In this case $\gamma = x_j D$.
    This case is quite similar to Case 1, with the difference being that Bob checks at the beginning if 
    \begin{align*}
     \size(\largeitems \cap \widetilde{F}_1) \geqslant D - \widetilde{\Delta}\,,
    \end{align*}
    i.e., without including $\widetilde{F}_2$ compared to~(\ref{eq:case1threshold}). 
    
    If the condition was indeed satisfied, then the same reasoning as the first part of Case 1 resolves this case. Otherwise, we get 
    \begin{align}
      \label{eq:nonnegsmallcase2}
      \size(\smallitems \cap \widetilde{F}_1) + x_j D > \apxsmallcontrib + (1-\delta) \apxrdemand\,,
    \end{align}
    and using Equation~(\ref{eq:case1expanded}) from the second part of Case 1 yields that that \emph{first four} terms in this case are non-negative and easy to compute. Similarly, adding and subtracting $(\apxsmallcontrib + (1-\delta) \apxrdemand)$ to the \emph{last four} terms of~(\ref{eq:case1expanded}), and rearranging the terms we get 
    \begin{align*}
     \left( 
   \size(\smallitems \cap \widetilde{F}_1) -\apxsmallcontrib - (1-\delta) \apxrdemand + x_j D\right) + 
    \left(\apxsmallcontrib + (1-\delta) \apxrdemand- \alpha U - \size(\smallitems \cap A)\right)\,.
    \end{align*} 
    The first part of the summation is non-negative by Equation~(\ref{eq:nonnegsmallcase2}) and
    can be computed by Bob. The second part is the same as the second part
    in Equation~(\ref{eq:secondpartofFCI}). It is non-negative by definition and can be computed by Alice.
    This completes the proof.
\end{description}
This concludes the promised communication problem in the case where Alice is given an infeasible tuple $(A,F_1,F_2)$, and Bob is given a canonical feasible solution with $\widetilde{F}_3=\emptyset$. As argued in Section~\ref{sec:fci:notation}, this generalizes to any canonical feasible solution \emph{without} any restriction on $\widetilde{F}_3$.
\subsection{Randomized Protocol for Arbitrary Feasible Solutions.}
\label{sec:fci:protocol}
We now extend the communication protocol of canonical feasible solutions to arbitrary feasible solutions. 
To that end, we denote by $\mathcal{R} = \{ (x^1,y^1),(x^2,y^2),\dots,(x^r,y^r)\}$ the set of all canonical feasible solutions. 

In this non-restricted setting, Alice still gets an infeasible tuple $(A,F_1,F_2)$, but Bob gets a feasible solution $(x,y)$ that is not necessarily canonical, and the goal remains to compute the slack of the corresponding flow-cover inequality (i.e., Equation~(\ref{eq:FCI})) in expectation. We show that the communication protocol that we developed in the previous section can be used as a black-box to handle this general case, by noting that any feasible solution $(x,y)$ can be written as a convex combination of canonical feasible solutions $(x^1,y^1), (x^2,y^2),\dots, (x^r, y^r)$. 
In other words, there exists $\lambda_1,\lambda_2,\dots, \lambda_r \geqslant 0$, $\sum_{k=1}^r \lambda_k = 1$, such that
\begin{align}
\label{eq:convexexpansion}
  (x,y) = \sum_{k=1}^r \lambda_k (x^k,y^k)\,.
\end{align}
This is formalized in Lemma~\ref{lem:canonical}.

To see that this is enough, note that the expansion in Equation~(\ref{eq:convexexpansion}) of $(x,y)$ allows us to rewrite slack of the flow-cover inequalities in~(\ref{eq:FCI}) as
\begin{align*}
    & \rsize( F_1 \cap B)  + \xsize (F_2 \cap B)  - \alpha U \\
  = ~& \sum_{i \in F_1} s_i\rq{} \sum_{k=1}^r \lambda_k y^k_i + \sum_{i \in F_2} \sum_{k=1}^r \lambda_k x^k_i D - \alpha U \\
  = ~& \sum_{k = 1}^r \lambda_k 
  {\left( \sum_{i \in F_1} s_i\rq{} y^k_i + \sum_{i \in F_2} x^k_i D - \alpha U \right)}\,.
\end{align*}
Thus in order to compute the slack in expectation, Bob 
samples a canonical feasible solution $(x^k,y^k) \in \mathcal{R}$ with probability $\lambda_k$,
then together with Alice, they compute the slack of \begin{align*}
\sum_{i \in F_1} s_i\rq{} y^k_i + \sum_{i \in F_2} x^k_i D - \alpha U
\end{align*}
as discussed in the previous section. 

It remains to prove that any feasible solution can indeed be written as a convex combination of canonical feasible solutions. This is formalized in Lemma~\ref{lem:canonical}.
\begin{Lemma}\label{lem:canonical}
Let $\mathcal{R} = \{ (x^1,y^1),(x^2,y^2),\dots,(x^r,y^r)\}$ be the set of all the canonical feasible solutions for the flow cover problem, then any feasible solution $(x,y)$ can be written as 
\begin{align*}
  (x,y) = \sum_{k=1}^r \lambda_k (x^k,y^k)\,,
\end{align*}
such that $\lambda_k \geqslant 0$ for all $1 \leqslant k \leqslant r$, and $\sum_k \lambda_k = 1$.
\end{Lemma}
\begin{proof}
Given a feasible solution $(x,y)$, define its support $\widetilde{F}^{x,y}=\{i: i \in F, \text{ and } y_i = 1\}$,
and define the set $\mathcal{R}^{x,y}$ to be the set of all canonical feasible solutions whose support equals $\widetilde{F}^{x,y}$, i.e., 
\begin{align*}
  \mathcal{R}^{x,y} = \{(x',y): (x', y) \in \mathcal{R}\} \subseteq \mathcal{R}\,.
\end{align*}
Without loss of generality, we assume that $\widetilde{F}^{x,y} = [n]$ to simplify the presentation. 

We now consider the following polytope $P(y)$:
\begin{align*}
P(y) = \left\{\begin{array}{ll}
        z \in [0,1]^n, & \text{such that: } \\
        (*) \, \, \sum_{i=1}^n z_{i} = 1, & \\
        (**)\, \, 0 \leqslant z_i \leqslant \frac{s_i y_i}{D} & \text{ for all } 1 \leqslant i \leqslant n
        \end{array}\right\} \\
\end{align*}
Note that for any feasible solution $(x,y)$ to the flow cover problem, we have that $x \in P(y)$. Moreover, we get from Definition~\ref{def:canonical} that for any canonical feasible solution $(x',y) \in \mathcal{R}^{x,y}$, all except at most one item $i \in [n]$, either has $x'_i =0$ or $x'_i D = s_i y_i$. Thus $x'$ satisfies
at least $n-1$ linearly independent constraints of type $(**)$ with equality. 
Conversely, if a point $x \in P(y)$ satisfies at least $n-1$ constraints of type $(**)$ with equality, then $(x,y) \in \mathcal{R}^{x,y}$.
 
Recall that  a point $z$ is an \emph{extreme point} solution of $P(y)$ iff there are $n$ linearly independent constraints that are set to equality by $z$. 
Since constraint $(*)$ is an equality constraint and is linearly independent from any set of $n-1$ constraints from $(**)$, we conclude that
$\{x': (x',y) \in \mathcal{R}^{x,y}\}$ is the set of all extreme points of $P(y)$. 
This implies that for any $x \in P(y)$, there exists $\lambda_k \geqslant 0$ for each $1 \leqslant k \leqslant r$ such that $\sum_{k} \lambda_k = 1$ and 
\begin{align*}
x = \sum_{k=1}^r \lambda_k x^k\,.
\end{align*}
Since all these points have the same $y$-support, it follows that \begin{align*}
(x,y) = \sum_{k=1}^{r} \lambda_k (x^k,y^k)\,.
\end{align*}
\end{proof}

\section{Algorithmic Aspects.} \label{sec:cutting_plane}

Theorem~\ref{thm:main} relies on the \emph{existence} of a quasi-polynomial size
extended formulation for the weakened knapsack cover inequalities. However, 
we do not know how to \emph{construct} the full extended formulation in 
quasi-polynomial time. Nevertheless, there is a way to use the extended
formulation algorithmically, which we describe here.

We adopt a more general point of view, since the findings of this section 
are applicable beyond the context of the knapsack cover inequalities. 
Consider any system of $p$ inequalities $A_1 x \geqslant b_1$, \ldots,
$A_p x \geqslant b_p$, and $q$ solutions $x^{(1)}$, \ldots,
$x^{(q)} \in \R^n$ of this system. In the context of the min-knapsack 
problem, the inequalities $A_i x \geqslant b_i$ ($i \in [p]$) are all the 
weakened knapsack cover inequalities and the solutions $x^{(j)}$ 
($j \in [q]$) are all the feasible solutions $x \in \{0,1\}^n$. Typically,
both $p$ and $q$ are exponentially large as functions of $n$.

To this data corresponds a slack matrix $S \in \R^{p \times q}_{\geqslant 0}$
defined by $S_{ij} := A_i x^{(j)} - b_i$. As observed by 
Yannakakis~\cite{Yannakakis91}, every non-negative factorization 
$S = FV$ where $F \in \R^{p \times r}_{\geqslant 0}$ and 
$V \in \R^{r \times q}_{\geqslant 0}$ determines a system 
\begin{align}
\label{eq:slack_param} A_ix-b_i &= F_iy \quad \forall i \in [p]\\ 
\nonumber y &\geqslant 0
\end{align}
whose projection to the $x$-space gives a polyhedron 
$\{x \in \R^n \mid \exists y \in \R^r : Ax-b = Fy,\ y \geqslant 0\}$ 
containing each of the solutions $x^{(j)}$ and contained in each of 
the halfspaces $A_i x \geqslant b_i$. 

Usually, the number $p$ of equations in~(\ref{eq:slack_param}) 
is much bigger than both the number $n$ of $x$-variables and 
rank $r$ of the non-negative factorization. Thus the equation 
system is largely overdetermined and can be replaced by a smaller 
equivalent subsystem with at most $n + r$ equations. However, it 
is not obvious to tell efficiently what are the indices $i$ for 
which the corresponding equation in~(\ref{eq:slack_param}) 
should be kept.

To avoid this difficulty, we assume that the way in which we want to 
use the extended formulation (shorthand: EF) $Ax - b = Fy$, $y \geqslant 0$
is to solve the LP $\min \{c^\intercal x \mid Ax \geqslant b\}$ 
for a \emph{given} objective vector $c \in \R^n$, through the 
extended formulation.

For $I \subseteq [p]$, consider the linear program
\begin{equation*}
\mathrm{LP}(I): \quad
\begin{array}[t]{r@{\ }r@{\ }ll}
\min &\multicolumn{1}{@{}l@{}}{c^\intercal x}\\
\mathrm{s.t.} &A_i x - b_i &= F_i y &\forall i \in I\\ 
&y &\geqslant 0\,.
\end{array}
\end{equation*}
In fact, we will only need to consider sets $I$ of size
at most $n + r \ll p$.

Algorithm~\ref{algo:cutting_plane} solves the LP 
$\min \{c^\intercal x \mid Ax \geqslant b\}$ in several 
steps. In each step, it solves the smaller $\mathrm{LP}(I)$ 
where $I \subseteq [p]$ and calls a separation routine to 
check whether $x^*$, the $x$-part of the optimum solution 
found, satisfies $Ax \geqslant b$ or not. In the first case, 
it returns $x^*$ and stops. In the second case, it adds the 
index $i^*$ of any violated constraint to $I$ and continues. 
At the beginning of the algorithm, $I$ is initialized to $[n]$.
To avoid technicalities, we assume that $\mathrm{LP}([n])$
is bounded. For the sake of concreteness, we assume furthermore 
that the $n$ first inequalities of the system $Ax \geqslant b$
are the nonnegativity inequalities $x_1 \geqslant 0$, \ldots,
$x_n \geqslant 0$, and that $c \in \R^n_{\geqslant 0}$. 

\begin{algorithm}[h]
\caption{Cutting-plane algorithm to solve $\min \{c^\intercal x \mid Ax \geqslant b\}$ through EF $Ax - b = Fy$, $y \geqslant 0$}
\label{algo:cutting_plane}
\begin{algorithmic}[1]
\STATE{initialize $I \longleftarrow [n]$}
\STATE{initialize feasible $\longleftarrow$ \textbf{false}}
\REPEAT
\STATE{solve $\mathrm{LP}(I)$, get optimum solution $(x^*,y^*)$}
\IF{there exists $i^* \in [p]$ such that $A_{i^*} x^* < b_{i^*}$}
\STATE{add $i^*$ to $I$}
\ELSE
\STATE{set feasible $\longleftarrow$ \textbf{true}}
\ENDIF
\UNTIL{feasible $=$ \textbf{true}}
\STATE{return $x^*$}
\end{algorithmic}
\end{algorithm}

To analyze the running time of the algorithm, we make the following
assumptions:

\begin{itemize}
\item the size of each coefficient in~(\ref{eq:slack_param}) and
each $c_i$ is upper-bounded by $\Delta = \Delta(n)$;
\item the separation problem (given $x^* \in \R^n$, find an index $i^* \in [p]$ such that $A_{i^*} x < b_{i^*}$ or report that no such index exists) can be solved in $\Tsep(n)$ time;
\item each single equation in~(\ref{eq:slack_param}) can be written down in $\Tconstr(n)$ time;
\item $\mathrm{LP}(I)$ can be solved in time $\Tsolve(n)$ for any set $I$
of size at most $n + r$, where $r = r(n)$ is the rank of the nonnegative
factorization giving rise to the extended formulation $Ax - b = Fy$, $y \geqslant 0$.
\end{itemize}

Notice that $\Tsolve(n) = O(n^3 (n+r) \Delta)$ if an interior point method 
is used to solve $\mathrm{LP}(I)$.

\begin{Lemma} \label{lem:cutting_plane_analysis}
Under the above assumptions, the main loop of 
Algorithm~\ref{algo:cutting_plane} is executed at most 
$r + 1$ times. Thus the complexity of Algorithm~\ref{algo:cutting_plane} 
is $O(r \cdot (\Tsolve(n) + \Tsep(n) + \Tconstr(n)))$. 
\end{Lemma}

\begin{proof}
The result follows directly from the simple observation that each time a 
new equation $A_{i^*} x - b_{i^*} = F_{i^*} y$ added to the system 
$A_{i} x - b_{i} = F_{i} y$ ($i \in I$), it is linearly independent 
from the current equations in the system. Notice that by assumption, 
the algorithm starts with $n$ linearly independent constraints. By the 
above observation, we always have $|I| \leqslant n + r$. 
\end{proof}

From now on, we assume that the non-negative factorization of the slack matrix $S$ comes from a communication protocol with non-negative outputs computing $S$ in expectation. The protocol is specified by a binary \emph{protocol tree}, in which each internal node is owned either by Alice or Bob, and each leaf corresponds to an output of the protocol. At each internal node $u$ owned by Alice, a \emph{branching probability} $\pbranch(i,u) \in [0,1]$ is given for each input $i \in [p]$ of Alice. Similarly for each internal node $v$ owned by Bob, we are given a branching probability $\qbranch(j,v) \in [0,1]$, where $j \in [q]$ is Bob's input. These branching probabilities specify the chance for the protocol of following the left branch. Finally, each leaf $\ell$ has a nonnegative number $\lambda(\ell) \in \R_{\geqslant 0}$ attached to it.

The corresponding extended formulation can be written as
\begin{align}
\label{eq:expected_output} A_i x - b_i &= \sum_{\ell\ \mathrm{leaf}} \preach(i,\ell) \cdot y_\ell \quad \forall i \in [p]\\ 
\nonumber y_\ell &\geqslant 0 \quad \forall \ell\ \mathrm{leaf}
\end{align}
where $\preach(i,u)$ denotes the probability of reaching node $u$ of the protocol tree on any input pair of the form $(i,*)$.

\begin{Lemma} \label{lem:single_constraint}
Let $\Delta$ be any number that is at least $\max \{-\log(\preach(i,\ell)) 
\mid i \in [p],\ \ell \textrm{ leaf },\ \preach(i,\ell) > 0\}$ and let 
$h$ denote the height of the protocol tree. For any fixed $i \in [p]$, 
one can write down the right-hand side of the corresponding equation in~(\ref{eq:expected_output}) in $O(2^h \Delta \log \Delta \log \log \Delta)$ 
time and $O(2^h \Delta)$ space.
\end{Lemma}

\begin{proof}
Clearly, at the root of the protocol tree, we have $\preach(i,\mathrm{root}) = 1$. At an internal node $u$ owned by Alice with left child $v$ and right child $w$, we have $\preach(i,v) = \preach(i,u) \cdot \pbranch(i,u)$ and $\preach(i,w) = \preach(i,u) \cdot (1-\pbranch(i,u)) = \preach(i,u) - \preach(i,v)$. In case $u$ is owned by Bob, we simply have $\preach(i,v) = \preach(i,w) = \preach(i,u)$ since the behavior of the communication protocol at node $u$ on input pair $(i,j)$ is independent of $i$.

Using this, we can compute recursively $\preach(i,u)$ for all nodes $u$ of the protocol tree, and thus for the leaves of the tree. All arithmetic operations are performed on numbers of at most $O(\Delta)$ bits. If we use the Schoolbook algorithm for subtraction and the Sch\"onhage-Strassen algorithm for multiplication, we obtain the claimed bounds for the time- and space-complexity.
\end{proof}

Now, we discuss how Algorithm~\ref{algo:cutting_plane} and its analysis
apply to the (weakened) knapsack cover inequalities and the corresponding 
slack matrix $(S^{\varepsilon}_{ab})_{a \in f^{-1}(0),\ b \in f^{-1}(1)}$
as in~(\ref{eq:relaxedKCIslack}), where $f$ is the weighted threshold
function defining the knapsack. In order to do that, we first have to 
construct the protocol tree of the protocol described in the proof of 
Lemma~\ref{lem:mainlemma}. We claim that this can be done in time
$(1/\varepsilon)^{O(1)} n^{O(\log n)}$. 

The protocol has several deterministic parts (in which the branching
probabilities are in $\{0,1\}$ locally). Each corresponds to the resolution 
of a Karchmer-Wigderson game. For writing down the corresponding subtrees 
of the protocol tree, we just need $\log^2(n)$-depth monotone circuits 
of fan-in $2$ for computing certain truncations of the weighted 
threshold function $f$. The translation of the circuit into a protocol 
tree follows the standard construction of Karchmer and 
Wigderson~\cite{KM90}. For constructing the circuits, we rely 
either on the construction of Beimel and Weinreb~\cite{BW05,BW06} 
or the simpler and more recent construction of Chen, Oliveira and 
Servedio~\cite{COS15}. Both constructions can be executed in $n^{O(1)}$ 
time. 

The remaining parts of the protocol can be readily translated
into the corresponding subtrees of the protocol tree.

Since the reaching probabilities in the protocol tree can be 
written down with $O(\log n)$ bits, each coefficient in the 
right-hand side of~(\ref{eq:expected_output}) can be written 
down in $O(\log n)$ bits. Assuming as before that all item 
sizes and demand can be written down with $O(n \log n)$ bits 
(which is without loss of generality), the coefficients of the
left-hand side of~(\ref{eq:expected_output}) can be written 
down with $O(n \log n)$ bits. Therefore, we can take 
$\Delta = O(n \log n)$ 

From what precedes and Lemma~\ref{lem:single_constraint}, 
we have that $\Tconstr(n) = (1/\varepsilon)^{O(1)} n^{O(\log n)}$.
Moreover, Lemma~\ref{lem:mainlemma} gives $r(n) = 
(1/\varepsilon)^{O(1)} n^{O(\log n)}$.

For the separation routine, we deviate significantly from
Algorithm~\ref{algo:cutting_plane}: instead of using an exact
separation routine (efficient exact separation of the knapsack 
cover inequalities is an open problem), we rely on a separation
trick from Carr \emph{et al.}~\cite{CFLP00}. That is, we
simply check if the knapsack cover inequality for $A := \{i \in [n] 
\mid x^*_i \geqslant 1/2\}$ is satisfied. This is enough to
guarantee that the modified Algorithm~\ref{algo:cutting_plane}
computes a quantity that is within a $2 + \varepsilon$ factor
of the integer optimum for that particular cost function $c$. 
Unfortunately, by relying on the pseudo-separation of 
Carr \emph{et al.}, we cannot guarantee that the modified
Algorithm~\ref{algo:cutting_plane} optimizes exactly over 
all weakened knapsack cover inequalities.

If we further assume that the coefficients of $c$ can be 
written with $O(n \log n)$ bits, we conclude that one can 
find a $(2+\varepsilon)$-approximation of $\min 
\{\sum_{i=1}^n c_i x_i \mid \sum_{i=1}^n \size_i x_i 
\geqslant \demand,\ x \in \{0,1\}^n\}$ in time 
$(1/\varepsilon)^{O(1)} n^{O(\log n)}$, without relying
on the ellipsoid algorithm, using our extended formulation.

\section{Conclusion.}

After the recent series of strong negative results on extended 
formulations, we have presented a positive result inspired by 
a connection to monotone circuits. Namely, we obtain the first 
quasi-polynomial-size LP relaxation of min-knapsack with constant 
integrality gap from polylog-depth circuits for weighted threshold
functions.

This result sheds new light on the approximability of min-knapsack 
via small LPs by connecting it to the complexity of monotone circuits.
For instance, it follows from our results that proving that no 
$n^{O(1)}$-size LP relaxation for min-knapsack can have integrality 
gap at most $\alpha$ for some $\alpha > 2$ would rule out 
the existence of $O(\log n)$-depth monotone circuits with bounded
fan-in for weighted threshold functions on $n$ inputs, which is 
an open problem.

Finally, let us further mention two open questions following this work. First,
it would be interesting to find an efficient (quasi-polynomial time) procedure
to explicitly write down our linear program for min-knapsack. Second, it would
be interesting to understand whether there is a ``combinatorial'' interpretation of
our relaxation.



\end{document}